\documentclass[11pt, a4paper,reqno]{amsart}
\usepackage{amsfonts}
\usepackage{amssymb}
\usepackage[dvips]{graphics}
\usepackage{epsfig}
\usepackage{euscript}
\usepackage{color}
\usepackage{enumitem}
\usepackage{fullpage,hyperref}
\usepackage{mathbbol}
\usepackage{enumitem} 
\usepackage{epigraph}
\usepackage{fancyvrb}

 \newtheorem{thm}{Theorem}[section]
 \newtheorem{cor}[thm]{Corollary}
 \newtheorem{lemma}[thm]{Lemma}
 \newtheorem{prop}[thm]{Proposition}
 \theoremstyle{definition}

    \newtheorem{conj}[thm]{Conjecture}
 \newtheorem{rem}[thm]{Remark}
 \newtheorem{ex}{Example}
 
 \numberwithin{equation}{section}

\newcommand\beq{\begin{equation}}
\newcommand\eeq{\end{equation}}
\newcommand{\pa}[1]{\left( #1 \right)}
\newcommand{\trip}{\|\kern-.08em |}

\newcommand{\caH}{{\mathcal H}}

\newcommand{\caL}{{\mathcal L}}

\newcommand{\caS}{{\mathcal S}}

\newcommand{\Z}{{\mathbb Z}}

\newcommand{\tr}{\mathrm{tr}}

\newcommand{\abs}[1]{\left\vert #1 \right\vert}

\newcommand{\norm}[1]{\left\Vert #1 \right\Vert}

\newcommand\be{\begin{equation}\begin{aligned}}
\newcommand\ee{\end{aligned}\end{equation}}
\newcommand{\bea}{\begin{eqnarray}}
\newcommand{\eea}{\end{eqnarray}}
\newcommand{\beann}{\begin{eqnarray*}}
\newcommand{\eeann}{\end{eqnarray*}}

\newcommand{\e}{\mathrm{e}}

\newcommand{\dist}{\mathrm{dist}}
\newcommand{\set}[1]{\left\{ #1 \right\}}

\author{Alexander Elgart}
\address{Department of Mathematics \\ Virginia Tech \\ Blacksburg, VA 24061-0123 \\ USA}
\email{aelgart@vt.edu}

\author{Martin Fraas}
\address{Department of Mathematics \\ UC Davis \\ Davis, CA 95616 \\ USA}
\email{fraas@vt.edu}

\begin{document}

\title{On Kitaev's determinant formula}
\epigraph{Sergey Naboko - in memoriam.}

\date{\today }

\begin{abstract} 
We establish a sufficient condition under which $\det\pa{ABA^{-1}B^{-1}}=1$ for a pair of bounded, invertible operators $A,B$ on a Hilbert space.
\end{abstract}
\thanks{Both authors are  supported in part by the NSF under grant DMS-1907435. A.E. is supported in part by the Simons Fellows in Mathematics grant.}
\maketitle

\section{Kitaev's formula and traces of certain commutators}
%
In a finite dimensional Hilbert space $\mathcal H$, the determinantal formula 
\be\label{eq:Det}
\det\pa{ABA^{-1}B^{-1}}=1
\ee
holds for any invertible operators $A,B\in\mathcal L(\mathcal H)$. Its naive generalization to the infinite dimensional case (via the Fredholm extension, see e.g., \cite[Sections 3]{S} for a background and basic properties) fails. A simple counterexample can be constructed using the Helton-Howe-Pincus formula, \cite{E}: Let $C$ and $D$ be bounded operators on a Hilbert space $\mathcal H$ such that $[C,D]\in \mathcal S_1$ (the Schatten trace class), where $[C,D]=CD-DC$ stands for the commutator of $C$ and $D$.  Then $\e^C\e^D\e^{-C-D}=I+S$, where  $I$ denotes the identity map  and $S\in\mathcal S_1$. In particular, the Fredholm operator below is well defined  and satisfies
\be\label{eq:PIN}\det\pa{\e^C\e^D\e^{-C-D}}=\e^{\frac12\tr[C,D]}.\ee
Thus ${\e^C\e^D\e^{-C}\e^{-D}}-I\in \mathcal S_1$ and using a basic property of the Fredholm determinant
\be\label{eq:Pinc}
\det{\pa{\e^C\e^D\e^{-C}\e^{-D}}}=\det{\pa{\e^C\e^D\e^{-C-D}}}\det{\pa{\e^{C+D}\e^{-C}\e^{-D}}}=\e^{\tr[C,D]}\ee
for such operators $C$ and $D$.


Let $R,L$ denote the forward and backward shift operators on $\ell^2(\mathbb N)$ (with respect to the standard basis $\set{e_n}$), and let $z\in\mathbb C$.  Then, the operators $A= \e^{zR}$, $B=\e^{L}$ are bounded and invertible, and moreover $[R,L]=P_1$, the orthogonal projection onto ${\rm Span}(e_1)$.  Hence, denoting by $I$ the identity map, \eqref{eq:Pinc} implies that $ABA^{-1}B^{-1}-I\in\mathcal S_1$ and 
\[\det\pa{ABA^{-1}B^{-1}}=\e^z,\]
i.e., the expression on the left hand side can take any non-zero complex value. 

It is therefore an interesting question to determine under which conditions \eqref{eq:Det} actually holds. Another motivation to study this object comes from physics, where it can be linked to the quantization of transport in quantum systems, \cite{K}. Indeed, if both \eqref{eq:Det} and \eqref{eq:Pinc} hold, one can deduce the quantization of $\tr[C,D]$, i.e., $\tr[C,D]\in2\pi i\Z$. Kitaev observed via a formal computation that, if a pair of unitaries $U_1=\e^{iC}$, $U_2=\e^{iD}$ with bounded self-adjoint operators $C,D$ satisfy $(U_1-I)(U_2-I),(U_2-I)(U_1-I)\in\mathcal S_1$, then \eqref{eq:Det} holds, implying the quantization for the case $[C,D]\in\mathcal S_1$. 


This suggest the following 
\begin{conj}\label{conj}
Let $\caH$ be a complex, separable Hilbert space. Suppose that 
\begin{enumerate}
\item $A,B\in \caL(\caH)$ are invertible;
\item $(A-I)(B-I),(B-I)(A-I)\in \mathcal S_1$.
\end{enumerate}
Then \eqref{eq:Det} holds.
\end{conj}
While we don't know how to prove Conjecture \ref{conj}, the purpose of this Note is to present an elementary derivation of the following weaker result.
\begin{thm}\label{thm:main}
Assume in addition to {\rm(i)-(ii)} that
\begin{enumerate}
\item[\rm{(iii)}]  $(A^*-I)(B-I),(B-I)(A^*-I)\in \mathcal S_1$.
\end{enumerate}
Then, \eqref{eq:Det} holds.
\end{thm}
Let us stress that the condition \rm{(iii)} is not a necessary, but only a sufficient condition. This can be seen from the following assertion (whose proof can be found in Section \ref{sec:proofs}).
\begin{prop}\label{lem3}
Let $C$ be a quasinilpotent operator  and $D$ bounded, such that
 \be\label{Kcond}\pa{\e^{C}-I}\pa{\e^{D}-I},\pa{\e^{D}-I}\pa{\e^{C}-I}\in\mathcal S^1.\ee
 Then 
 \[\det\pa{\e^C\e^{D}\e^{-C}\e^{-D}}=1.\]
\end{prop}
We now construct a simple example, that shows that condition \rm{(iii)} is not a necessary, based on this observation.
\begin{ex}
Let $C=D=ML$, where $L$ is the backward shift operator on $\ell^2(\mathbb N)$, and $M$ is a multiplication operator on the same space defined by 
\[Me_n=\begin{cases}\frac1{\sqrt n} e_n& n\in2\mathbb N\\ 0 & n\in2\mathbb N-1\end{cases}.\]
Then \eqref{eq:Det} holds trivially for $A=\e^C$, $B=\e^D$ as $C,D$ commute. We also note that  $C^2=0$ (so $C$ is nilpotent) and  $\e^C-I=C$ (so \eqref{Kcond} holds as well). However, $\pa{\e^C-I}\pa{\e^{C^*}-I}=CC^*=M^2\notin\mathcal S_1$, so (iii) in Theorem \ref{thm:main} is not satisfied.  
\end{ex}

\begin{rem}
We next note that if $A$ (or $B$) is normal, then \rm{(ii)} is equivalent to  \rm{(iii)}, so in this case Conjecture \ref{conj} becomes a theorem, confirming Kitaev's formal observation. In fact, the proofs of Theorem \ref{thm:main} and Proposition \ref{lem3} can be combined to show that Conjecture \ref{conj} is satisfied for the so-called spectral operators, introduced by Dunford, \cite{D}. 
\end{rem}
As we have already mentioned, \eqref{eq:Pinc} immediately implies
\begin{cor}\label{cor}
If $C,D\in\mathcal L(\mathcal H)$ satisfy \eqref{Kcond} and $[C,D]\in\mathcal S^1$, then $\tr[C,D]\in  2\pi i \Z$.
\end{cor}
One can of course suspect, based on the vanishment of the trace of the commutator in the finite dimensional case, that in fact the only allowed value for 
$\tr[C,D]$ in the statement above is zero. To this end, we  construct 
\begin{ex}\label{ex}
There exist self-adjoint operators $C,D$ satisfying the assumptions of Corollary \ref{cor}  such  that  $\tr[C,D]=-8\pi i $. Specifically, let $C=f(x):= 2\pi i \frac{x}{\langle x\rangle}$ where $\langle x\rangle=(1+x^2)^{1/2}$ and let $D=f(p)$, where $x$ and $p=-i\frac d{dx}$ are the position and momentum operators on $L(\mathbb R)$, see \cite[Section 4]{S} for details (we note that here $f(p)$ is understood as a convolution operator, see \cite[Theorem IX.29]{RS}).  
Then  \eqref{Kcond} is satisfied, $[C,D]\in\mathcal S^1$, 
and $\tr[C,D]=-8\pi i$.
\end{ex}
We will verify the validity of this construction at the end of Section \ref{sec:proofs}. 

 Since $U^k-I=(U-I)\sum_{j=0}^{k-1}U^j$ for any unitary $U$ and any $k\in\Z$, one deduces from this example that there are operators $C,D$ satisfying Corollary \ref{cor} above such that $\tr[C,D]=8k\pi i $ for {\it any} $k\in\Z$.
\section{Proofs}\label{sec:proofs}
\begin{proof}[Proof of Theorem \ref{thm:main}]
\begin{lemma}\label{lem:SVD}
Assume that the assumptions of Theorem \ref{thm:main} hold.  Let $A=U|A|$ be the polar decomposition for $A$. Then $U$ is in fact a unitary operator, $|A|$ is invertible, there are $C,D$ that are normal and bounded such that $|A|=\e^{C}$, $U=\e^{D}$, and we have
\be\label{eq:SVD}
\pa{|A|-I}(B-I),(B-I)\pa{|A|-I},\pa{U-I}(B-I),(B-I)\pa{U-I} \in \mathcal S_1.
\ee
In addition, the formula 
\be\label{eq:prod}
\det\pa{ABA^{-1}B^{-1}}=\det\pa{|A|\,B\,|A|^{-1}B^{-1}}\det\pa{BU^*B^{-1}U}
\ee
holds.
\end{lemma} 
\begin{proof}
The fact that $U$ and $|A|$ are invertible (and consequently have exponential representation in terms of normal operators) follows directly from the invertibility of $A$, so we only need to establish \eqref{eq:SVD}--\eqref{eq:prod}. To this end, we note that 
\[(A^*A-I)(B-I)=(A^*+I)(A-I)(B-I)-(A-I)(B-I)+(A^*-I)(B-I) \in \mathcal S_1\]
by {\rm(ii-iii)}. Hence
\[(|A|-I)(B-I)=(|A|+I)^{-1}(A^*A-I)(B-I)\in \mathcal S_1\]
as well. An identical argument yields the inclusion $(B-I)(|A|-I)\in \mathcal S_1$. We also have
\begin{multline*}
(B-I)(U-I)=(B-I)(A-|A|)|A|^{-1}\\ =(B-I)(A-I)|A|^{-1}-(B-I)(|A|-I)|A|^{-1}\in \mathcal S_1.
\end{multline*}
Finally, we have
\begin{multline*}
(U-I)(B-I)=(A-|A|)|A|^{-1}(B-I)=(A-|A|)(B-I)+(A-|A|)(|A|^{-1}-I)(B-I)\\=(A-I)(B-I)-(|A|-I)(B-I)-(A-|A|)|A|^{-1}(|A|-I)(B-I)\in \mathcal S_1,
\end{multline*}
so we established \eqref{eq:SVD}.

The relation \eqref{eq:prod} follows from the fact that $|A|B|A|^{-1}B^{-1}=I+K$, $BU^*B^{-1}U=I+M$ with $K,M\in \mathcal S_1$ by  
\be\label{eq:comform}
ABA^{-1}B^{-1}=I+[A,B]A^{-1}B^{-1}
\ee 
 and \eqref{eq:SVD}, the representation
\[ABA^{-1}B^{-1}=U\pa{|A|B|A|^{-1}B^{-1}}\pa{BU^*B^{-1}U}U^*,\]
as well as the basic properties of the Fredholm determinant.
\end{proof}
Applying Lemma \ref{lem:SVD} twice, we see that the statement of Theorem \ref{thm:main}  follows from
\begin{prop}\label{thm}
Let $A,B$ be  bounded normal operators in $\mathcal H$ that satisfy
 \[\pa{\e^{A}-I}\pa{\e^{B}-I},\pa{\e^{B}-I}\pa{\e^{A}-I}\in\mathcal S_1.\]
  Then  
\[\det\pa{\e^{A}\e^{B}\e^{-A}\e^{-B}}=1.\]
\end{prop}
\end{proof}
\begin{proof}[Proof of Proposition \ref{thm}]
We will use the following:
\begin{lemma}\label{lem2}
Let $B,D$ be a pair of bounded operators on $\mathcal H$ that satisfy
\[D\pa{\e^{B}-I},\pa{\e^{B}-I}D\in\mathcal S_1.\]
Then 
\[\det\pa{\e^{D}\e^{B}\e^{-D}\e^{-B}}=1.\]
\end{lemma}
\begin{proof}
Using a basic property of the Fredholm determinant, we have 
\be\nonumber
\det\pa{\e^{D}\e^{B}\e^{-D}\e^{-B}}&=\det\pa{\e^{B}\e^{-D}\e^{-B}\e^{D}\e^{ \pa{\e^{B}D\e^{-B}-D}}\e^{- \pa{\e^{B}D\e^{-B}-D}}}\\ &=\det\pa{\e^{- \e^{B}D\e^{-B}}\,\e^{D}\e^{ \pa{\e^{B}D\e^{-B}-D}}}\det\pa{\e^{- \pa{\e^{B}D\e^{-B}-D}}},
\ee
where both determinants on the right hand side are well-defined. 
We now use \eqref{eq:PIN} to evaluate the first determinator on the right hand side:
\[\det\pa{\e^{- \e^{B}D\e^{-B}}\,\e^{D}\e^{\pa{\e^{B}D\e^{-B}-D}}}=\exp\pa{-\tfrac12\tr\left[\e^{B}D\e^{-B},D\right]}=1,\]
since
\[\tr\left[\e^{B}D\e^{-B},D\right]=\tr\left[\e^{B}D\e^{-B}-D,D\right]=0,\]
where in the last step we used $\e^{B}D\e^{-B}-D=[\e^{B}-I,D]\e^{-B}\in \mathcal S_1$. 
We recall a consequence of Lidskii's theorem: If $X,Y\in\caL(\caH)$ are such that $XY,YX\in\caS_1$, then $\tr\pa{XY}=\tr\pa{YX}$.  Thus
\be\label{eq:inv}
\tr\pa{\e^{B}D\e^{-B}-D}=\tr\pa{\pa{\e^{B}-I}\pa{D\e^{-B}}}-\tr\pa{ \pa{D\e^{-B}}\pa{\e^{B}-I}} =0,
\ee
Using $\det\pa{\e^E}=\e^{\tr E}$ for $E\in\mathcal S_1$ and  \eqref{eq:inv}, we  get 
\[\det\pa{\e^{- \pa{\e^{B}D\e^{-B}-D}}}=1.\]

\end{proof}
Let  $P$ be the  spectral projection  $\chi_{2\pi i \Z}(A)$, where $\chi_W$ stands for the indicator of a set $W$.  
  Then, $\e^{AP}=I$  for a normal operator $A$ and $\det\pa{\e^{AP}\e^{B}\e^{-AP}\e^{-B}}=1$. 
%
%
%
Let $\Delta\in(0,1/2]$, let $W=\set{x\in\mathbb C:\ \dist\pa{x,2\pi i \Z}\ge\Delta}$, and let $Q_\Delta=\chi_W(A)$, and let $P_\Delta=I-Q_\Delta$. 

We first observe that since $\e^{A}-I$ is invertible on $Range\pa{Q_\Delta}$, we have
\[Q_\Delta\pa{\e^{B}-I}=\pa{\pa{\e^{A}-I}^{-1}Q_\Delta}\pa{\e^{A}-I}\pa{\e^{B}-I}\in\mathcal S_1,\]
and similarly
\[\pa{\e^{B}-I}Q_\Delta\in\mathcal S_1.\]
Thus, by Lemma \ref{lem2} we deduce 
\be\label{eq:Pinc1}\det\pa{\e^{AQ_\Delta}\e^{B}\e^{-AQ_\Delta}\e^{-B}}=1.
\ee
Next, we note that 
\[\left[\e^{AP_\Delta},\e^{B}\right]=\left[P_\Delta\pa{\e^{A}-I},\pa{\e^{B}-I}\right]\underset{\Delta\to0}{\to} \left[P\pa{\e^{A}-I},\pa{\e^{B}-I}\right]=\left[\e^{AP},\e^{B}\right],\]
where the convergence is in the trace norm sense (this follows from $SOT-\lim P_\Delta=P$ and the assumptions of Proposition \ref{thm}). This implies 
\be\label{eq:Pinc2}
\det\pa{\e^{AP_\Delta}\e^{B}\e^{-AP_\Delta}\e^{-B}}=\det\pa{I+\left[\e^{AP_\Delta},\e^{B}\right]\e^{-AP_\Delta}\e^{-B}} \to\det\pa{\e^{AP}\e^{B}\e^{-AP}\e^{-B}}=1
\ee
as $\Delta\to0$. 
We now can combine \eqref{eq:Pinc1} and \eqref{eq:Pinc2} to get
\be\label{eq:Pinc3}
\det\pa{\e^A\e^{B}\e^{-A}\e^{-B}} =\det\pa{\e^{AP_\Delta}\e^{B}\e^{-AP_\Delta}\e^{-B}}\det\pa{\e^{B}\e^{-AQ_\Delta}\e^{-B}\e^{AQ_\Delta}} \underset{\Delta\to0}{\to} 1.
\ee
\end{proof}

\begin{proof}[Proof of Proposition \ref{lem3}]
The statement  follows from
\be\label{eq:trcl}
C\pa{\e^{B}-I},\pa{\e^{B}-I}C\in\mathcal S_1
\ee 
and   Lemma \ref{lem2}.

To show \eqref{eq:trcl}, we observe that, denoting
\[D:=\sum_{k=0}^\infty\frac{C^k}{(k+1)!},\]
 we have $\e^{C}-I=DC$. Hence, \eqref{eq:trcl} will follow provided that $D$ is invertible, as
\[C\pa{\e^{B}-I}=D^{-1}\pa{\e^{C}-I}\pa{\e^{B}-I}. \]
To prove that $D$ is invertible, it suffices to show that $D-I$ is a (quasi)nilpotent operator. To this end, we can bound
\[\norm{\pa{I-D}^n}=\norm{(CE)^n}\le\norm{C^n}\norm{E}^n,\quad E=\sum_{k=0}^\infty\frac{C^k}{(k+2)!}.\]
We have
\[\norm{E}\le \sum_{k=0}^\infty\frac{\norm{C}^k}{(k+2)!}\le \e^{\norm{C}},\]
so
\[\norm{\pa{I-D}^n}^{1/n}\le\norm{C^n}^{1/n}\e^{\norm{C}}\to0\mbox{ as } n\to\infty,\]
so $D-I$ is indeed (quasi)nilpotent, and we are done. 
\end{proof}

\begin{proof}[Verification of Example \ref{ex}]
We first note that the conditions \eqref{Kcond} are satisfied by \cite[Theorem XI.21]{RS1} since there exists a $C>0$ such that 
\[\abs{\e^{{f(x)}}-1}\langle x\rangle^2\le C.\] 

We will use the integral representation
\[\frac\pi2\frac1{\langle p\rangle}=\int_0^\infty \frac{dt}{p^2+1+t^2},\]
which implies
\be\nonumber
[C,D]&=[C,p]\frac{2\pi i}{\langle p\rangle}+p\left[C,\frac{2\pi i}{\langle p\rangle}\right]\\ &=-f'(x)\frac{2\pi}{\langle p\rangle}+4\int_0^\infty \frac{p}{p^2+1+t^2}\pa{f'(x)p+pf'(x)}\frac{1}{p^2+1+t^2}dt.
\ee
The integral  can be written as 
\be\nonumber 
\int_0^\infty f'(x)\frac{2p^2}{\pa{p^2+1+t^2}^2}dt&-\int_0^\infty \left[f'(x),\frac{p}{p^2+1+t^2}\right]\frac{p}{p^2+1+t^2}dt\\&-\int_0^\infty \left[f'(x),\frac{p^2}{p^2+1+t^2}\right]\frac{1}{p^2+1+t^2}dt.
\ee
We note that the integrands in the second and third terms have trace norms decaying faster than $\frac{1}{t^2+1}$ in $t$, in particular these terms are trace class, see, e.g., \cite[Section 4]{S} for the trace class properties of the products of functions $F(x)G(p)$. Hence, we get
\[[C,D]=-4f'(x)\int_0^\infty \pa{\frac{1}{p^2+1+t^2}-\frac{2p^2}{\pa{p^2+1+t^2}^2}}dt + \mathcal T=-if'(x)f'(p)+ \mathcal T,\]
where $ \mathcal T$ is a trace class operator. Since $f'(x)=\frac{2\pi i}{\langle x\rangle^3}$, we see that $f'(x)f'(p)\in\mathcal L^1$, so $[C,D]\in\mathcal S_1$ as well. In fact, $\tr\,{\mathcal T}=0$ (this term originates from the commutator of $f'(x)$ with functions of $p$ that decay in $p$ sufficiently fast), so 
\[\tr\,{[C,D]}=i\tr\,{ f'(x)f'(p)}=\frac i{2\pi} \pa{\int_{\mathbb R} f'}^2=-8\pi i,\]
where in the second step we have used the fact that $f'(p)$ is a convolution operator,
\[(f'(p)\phi)(x)=(2\pi)^{-1/2}\int \check{(f')}(x-y)\phi(y)dy,\quad  \check{(f')}(x):=(2\pi)^{-1/2}\int \e^{ixp}{(f')}(p)dp,\]
see  \cite[Theorem IX.29]{RS}. 
\end{proof}
\section*{Acknowledgment}
We thank the reviewers for many suggestions used to improve the manuscript.

\end{document}